\documentclass[preprint]{imsart}

\RequirePackage[OT1]{fontenc}
\RequirePackage{amsthm,amsmath}
\RequirePackage[numbers]{natbib}
\RequirePackage[colorlinks,citecolor=blue,urlcolor=blue]{hyperref}
\usepackage{graphicx}
\usepackage[english]{babel}
\usepackage[T1]{fontenc}
\usepackage{amsmath}
\usepackage[utf8]{inputenc}
\usepackage{amssymb}
\usepackage{graphicx}
\usepackage{amscd}

\usepackage{tikz}
\usetikzlibrary{trees}
\tikzset{
  treenode/.style = {align=center, inner sep=0pt, text centered,
    font=\sffamily},
  arn_n/.style = {treenode, circle, white, font=\sffamily\bfseries, draw=black,
    fill=black, text width=0.5em},
    arn_n1/.style = {treenode, circle, black, font=\sffamily\bfseries, draw=black,
    fill=white, text width=1em},
  arn_r/.style = {treenode, circle, red, draw=red, 
    text width=1.5em, very thick},
  arn_x/.style = {treenode, rectangle, draw=black,
    minimum width=0.1em, minimum height=0.1em}
}

\makeatletter \def\@biblabel#1{} \makeatother

\newtheorem{definition}{Definition}[section]

\newtheorem{proposition}[definition]{Proposition}
\newtheorem{fact}[definition]{Fact}

\newtheorem{remark}{Remark}

\newtheorem{example}{Example}

\usepackage{color}

\newcommand{\ind}[1]{1\!\!1_{#1}}

\def\a {\alpha}

\def\o {\omega}
\def\O {\Omega}
\def \veps {\varepsilon}
\def \vphi {\varphi}

\def\CL { \mathcal{L}}

\def\CF {\mathcal{F}}

\def\CM {\mathcal{M}}
\def\CL {\mathcal{L}}
\def\CB {\mathcal{B}}
\def\CS {\mathcal{S}}
\def\CR {\mathcal{R}}
\def\CT { \mathcal{T}}

\def\CU {\mathcal{U}}

\def\CR {\mathcal{R}}

\def\E {\mathbb{E}}

\def\P {\mathbb{P}}
\def\R {\mathbb{R}}

\def\I {\mathbb{I}}

\def\tl {\tilde{l}}
\def\tr {\tilde{r}}

\def\ti {\tilde{i}}

\def\tnu {\tilde{\nu}}
\def\tbe {\tilde{\beta}}

\def\ti {\tilde{i}}
\def\tX {\tilde{X}}

\def\tnu {\tilde{\nu}}
\def\tbe {\tilde{\beta}}

\def\lI {l_{{\rm I}}}
\def\rI {r_{{\rm I}}}
\def\lII {l_{{\rm II}}}
\def\rII {r_{{\rm II}}}

\def\tlI {\tilde{l_{{\rm I}}}}
\def\trI {\tilde{r_{{\rm I}}}}
\def\tlII {\tilde{l_{{\rm II}}}}
\def\trII {\tilde{r_{{\rm II}}}}

\def\AI {A_{{\rm I}}}
\def\BI {B_{{\rm I}}}
\def\AII {A_{{\rm II}}}
\def\BII {B_{{\rm II}}}

\begin{document}
\begin{frontmatter}
\title{Inequality and risk aversion in economies open to altruistic attitudes}


\begin{aug}
\author{\fnms{Eleonora} \snm{Perversi}\thanksref{t1}\ead[label=e1]{eleonora.perversi@unipv.it}}
\and
\author{\fnms{Eugenio} \snm{Regazzini}\thanksref{t2}\ead[label=e2]{eugenio.regazzini@unipv.it}}

\thankstext{t1}{Work partially supported by MIUR-2012AZS52J-003 
}
\thankstext{t2}{Also affiliated with CNR-IMATI, Milano, Italy. Work partially supported by MIUR-2012AZS52J-003.
}
\runauthor{E. Perversi and E. Regazzini}

\affiliation{Universit\`a degli Studi di Pavia}

\address{Universit\`a degli Studi di Pavia, Dipartimento di Matematica, via Ferrata 1
27100 Pavia, Italy\\
\printead{e1}\\
\phantom{E-mail:\ }\printead*{e2}}

\end{aug}

\begin{abstract}
This paper attempts to find a relationship between agents' risk aversion and inequality of incomes. Specifically, a model is proposed for the evolution in time of surplus/deficit distribution, and the long-time distributions are characterized almost completely. They turn out to be weak Pareto laws with exponent linked to the relative risk aversion index which, in turn, is supposed to be the same for every agent. On the one hand, the aforesaid link is expressed by an affine transformation. On the other hand, the level of the relative risk aversion index results from a frequency distribution of observable quantities stemming from how agents interact in an economic sense. Combination of these facts is conducive to the specification of qualitative and quantitative characteristics of actions fit for the control of income concentration.
\end{abstract}




\end{frontmatter}

\section{Introduction, hypotheses and main results}\label{sec:intro}
In recent years the debate on inequality has been vivified by the appearance of writings that have attracted, and continue to attract, the attention of the general public. Suffice it to mention books by Atkinson (2015), Piketty (2013), Stiglitz (2013), along with the echo these works have found, and continue to find, in the public opinion. See also OECD (2015). The clear message coming from them is that inequality, far from being a necessary condition for good economies, may represent a serious bar to a smooth running of the economics of any community, to say nothing of a well-known experimental evidence of a propensity to inequality-aversion along with a certain kind of social preferences including egalitarianism. The present paper aims at answering the following two questions, within a clearly defined framework: First, has an inequality-averse community to resign itself to coexisting with low-level performance of its economic system or, on the contrary, is inequality aversion conducive to incentives suitable for reducing risk-aversion and, as a consequence, to improve performance of an economic system? Second, in this last case, how might feasible political actions influence individual risk-aversion? Answering these questions is not a mere trifle, so that we here confine ourselves to considering inequality (concentration) in the distribution of individual \textit{s-plus}, i.e., in the distribution of the difference between total production of wealth and subsistence, where subsistence is meant as wealth necessary to keep producers alive and cover the long-term cost of production. The s-plus corresponds to a \textit{surplus} or a \textit{subplus} (=\textit{deficit}) depending on the sign of the aforesaid difference. Thus, the present paper deals neither with the distribution of income among production factors, nor with the consumption of basic necessities. It will be assumed that the formation of the distribution of s-plus is determined by a large number of contracts between agents. Any of such contracts will be characterized by two vectors of (observable) non-negative ratios (see \eqref{interazioni} in Subsection \ref{sec:SUBassumptions}) to incorporate the idea that, in general, the gain (both positive or negative) of each agent is the sum of an s-plus derived from her/his initial s-plus, and of an s-plus she/he might extract from the initial s-plus of the other party of the contract. It is worth warning the reader that ``to extract'' has here to be meant in such a way that an agent might, directly of vicariously, subsidize the other party if this has an initial deficit, in agreement with the supposition we are dealing with a community of agents exhibiting social preferences including altruism. As to the notion of contract adopted in the present work, it is worth stressing that it is quite ``irregular'' with respect to the almost entire impressive economic literature on the subject. In point of fact, in modelling contracts, economists,  although they manifest a variety of discordant opinions, agree in assuming that all agents are pure selfish persons, whose aspiration consists in maximizing a utility function, possibly under more or less stringent constraints. Economists indicate conditions under which equilibria can be reached, supposing that agents are all in a position to tend towards such equilibria along ways without obstacles. Then, it is no wonder if the reaction to every proposal of such a nature is that such-and-such an assumption is implausible. This is often considered as a seemingly good reason to restart with some variant of the story which is, however, doomed to the same fate as its predecessor, with very high probability. A few valuable hints on how to get out of this vicious circle can be found, for example,
in the work by Rodolfo Benini (1862-1956), an Italian economist, demographer and statistician inclined to understand the fundamentals of the socialist theory of surplus, and elaborate them according to his personal bent towards an original inductive approach to economics. An important aspect of his oeuvre is the analysis of the fundamental role played by the \textit{resistance capacity} of the parties to a contract in determining the splitting up of the benefits of the ensuing exchange between the agents. To put it in a nutshell, the richer has wider freedom in choosing between present and future goods or services, between goods or services from the same market or from different markets, etc.. Such a person is in general in a position to wait the ``capitulation'' of the competitor, and the wait will be, in general, profitable for him. On the contrary, who possesses scant resources cannot wait. The value she/he attaches to what she/he can offer (she/he needs, respectively) decreases (increases, respectively) as the competitor delay in accepting the offer gets longer (as her/his resistance to the lack of the requirements gets weaker, respectively). See Benini (1929, 1938). This is enough to explain the leading role of the initial positions of agents in determining the effects of a contract, and the consequent description of these effects adopted in the present work, that is, in terms of mere accounting of the aforesaid distribution of advantages as functions of the initial positions. See next formula \eqref{interazioni}. On the basis of this naif, but realistic, stance, combined with a few hypotheses concerning both the social evaluation of the outcomes of any contract and the dynamics of the s-plus determined by the contracts between agents (Subsection \ref{sec:SUBassumptions}), we set about answering (Subsection \ref{sec:main} and Section \ref{sec:2}) the questions formulated at the beginning of these preliminary considerations.

\subsection{Discussion of the main assumptions}\label{sec:assumptions}
For a better understanding of the results of the present paper a reasoned description of the leading assumptions is now made. This requires a preliminary explanation of two concepts continuously surfacing in the rest of the paper: First, the concentration of a transferable (statistical) character, that is, in other terms, the inequality of its distribution among statistical units. Second, agents risk aversion. These notions are here conceived from an operational point of view and, then, they are recalled by means of well-defined measuring procedures in order to verify practically any statement where these notions are mentioned. See Bridgman (1927).

\subsubsection{Concentration}\label{sec:subsecConc}
The phrase ``concentration of a transferable character'' is used by statisticians with the same meaning as ``inequality in the distribution of that character''. Thus, a distribution is said to be less concentrated than another one when it is more egalitarian than the latter. These vague propositions can be made meaningful by resorting to the \textit{Lorenz concentration function} extended to arbitrary probability distributions (p.d.'s, for short) on $\R$, even with negative values in their support. For the classical usual definition of such a function, see, for example, Paragraph 2.25 in Stuart and Ord (1994). As for the aforesaid extension, let $\CF_0$ denote the class of p.d. functions on $\R$ with finite expectation, deprived of the one corresponding to the unit mass at $0$, which, besides, would be of no interest from the point of view of concentration. Then, for every $F$ in $\CF_0$, define $A_F$ to be the p.d. function of the absolute value of any random number distributed according to $F$, i.e.
\[
A_F(x)=\Big(F(x)-F(-x-0)\Big)\ind{[0,+\infty)}(x)\qquad(x\in\R).
\]
Moreover, set $A_F^{-1}(t)=\inf\{x\in\R:\;A_F(x)> t\}$ for any $t$ in $(0,1)$. Writing $\CR(A_F)$ for the range of $A_F$ and $\CM(A_F)$ for the expectation $\int_0^{+\infty}x d A_F(x)$, the \textit{extension} to $[0,1]$, by linear interpolation, of the function 
\[
\CR(A_F)\ni\theta\mapsto\dfrac{1}{\CM(A_F)}\int_0^\theta A_F^{-1}(t)dt
\]
is just the classical Lorenz concentration function of $A_F$. In the present work, it is used as concentration function for $F$ in $\CF_0$, denoted by $\vphi_F$. In fact, for any pair of elements $F_1$, $F_2$ of $\CF_0$ such that $\vphi_{F_1}\leq \vphi_{F_2}$, with strict inequality for some point of $(0,1)$, $A_{F_2}$ can be obtained from $A_{F_1}$  by redistributing absolute value of s-plus from higher absolute value to lower one. With reference to the original p.d. function, this is tantamount to saying that $F_2$ can be obtained from $F_1$ by a policy of equalizing transfers. Hence, consistently with the \textit{Dalton principle of transfers}, $F_2$ can be viewed as \textit{less concentrated than} $F_1$. See Dalton (1920) and, for further developments, Pratt (1964), Atkinson (1970), Rothschild and Stiglitz (1970, 1971 and 1972), Cifarelli and Regazzini (1987), Regazzini (1992).

\subsubsection{Relative risk aversion index}\label{sec:subsecRRAI}
The notion of risk aversion has been derived $-$ by Bruno de Finetti (1906-1985) $-$ from the classical expected utility principle, which dates back to Daniel Bernoulli and agrees with the axioms proposed by von Neumann and Morgensten (1944). Thus, it is assumed that each individual has a \textit{utility function} of money (a strictly increasing and continuous real-valued function on $\R$) $\bar{u}$ so that with any random gain $G$, having p.d. function $F$, it is associated a non-random gain $\hat{g}$, equivalent to $G$, in the sense that
\[
\int_{\R}\bar{u}(x)dF(x)=u(\hat{g}).
\]
Consider now the problem of the certainty equivalent gain for the restriction $F_0$ of $F$ to a small neighbourhood $N_0:= (x_0-\veps, x_0+\veps]$ of any arbitrary point $x_0$ in the support of $F$:
\[
F_0(x):=\left\{
		\begin{aligned}
		& 0 & \text{ if } & x <x_0-\veps\\
		& \dfrac{F(x)-F(x_0-\veps)}{F(x_0+\veps)-F(x_0-\veps)} & \text{ if } & x_0-\veps \leq x < x_0\\
		& 1 & \text{ if } & x \geq x_0+\veps
		\end{aligned}
		\right.
\]
Then, the certainty equivalent gain $\hat{g}_0$ of a random gain distributed according to $F_0$ has to satisfy
\[
\int_{N_0}\bar{u}(x)dF_0(x)=\bar{u}(\hat{g}_0).
\]
By an elementary estimation of the first member, which becomes better and better as $\veps$ decreases, one has
\[
\bar{u}(\hat{g}_0)+ \int_{N_0}(x-\hat{g}_0)\bar{u}'(\hat{g}_0)dF_0(x)+\dfrac{1}{2}\int_{N_0}(x-\hat{g}_0)^2\bar{u}''(\hat{g}_0)dF_0(x)\simeq \bar{u}(\hat{g}_0) 
\]
or, equivalently, 
\[
\dfrac{\int_{N_0}(x-\hat{g}_0)dF_0(x)}{\int_{N_0}(x-\hat{g}_0)^2dF_0(x)} \simeq -\dfrac{\bar{u}''(\hat{g}_0)}{2\bar{u}'(\hat{g}_0)}  \simeq -\dfrac{\bar{u}''(x_0)}{2\bar{u}'(x_0)} =: I_{\bar{u}}(x_0).
\]
It is easy to interpret the right hand side as a local \textit{risk aversion index}. Indeed, one can rewrite the above equation as
\[
\dfrac{m_0-\hat{g}_0}{\sigma_0^2+ (m_0-\hat{g}_0)^2}= \dfrac{\lambda_{\bar{u}}(x_0)}{x_0}
\]
where $m_0:= \int_{N_0}x dF_0(x)$, $\sigma_0^2:=\int_{N_0}(x-m_0)^2 dF_0(x)$ and 
\begin{equation}\label{lambdarai}
\lambda_{\bar{u}}(x_0):= -x_0 \dfrac{\bar{u}''(x_0)}{2\bar{u}'(x_0)}.
\end{equation}
Hence, one gets
\[
m_0-\hat{g}_0=\dfrac{x_0 \pm \sqrt{x_0^2-4 \sigma_0^2\lambda^2_{\bar{u}}(x_0)}}{2\lambda_{\bar{u}}(x_0)}.
\]
For the most commonly adopted utility functions $\bar{u}$, $\lambda_{\bar{u}}(x_0)$ turns out to be bounded from above by a suitable constant. Combination of this fact with the obvious inequality $\sigma_0^2\leq 4\veps^2$, yields $\sqrt{x_0^2-4 \sigma_0^2\lambda^2_{\bar{u}}(x_0)}\simeq |x_0|$ for sufficiently small $\veps$. Then, since $|m_0-\hat{g}_0|\leq 2\veps$, there is only one admissible value for $m_0-\hat{g}_0$, that is
\[
m_0-\hat{g}_0=\dfrac{x_0 -sign(x_0) \sqrt{x_0^2-4 \sigma_0^2\lambda^2_{\bar{u}}(x_0)}}{2\lambda_{\bar{u}}(x_0)}
\]
which gives
\[
\hat{g}_0=m_0-\dfrac{x_0 -sign(x_0) \sqrt{x_0^2-4 \sigma_0^2\lambda^2_{\bar{u}}(x_0)}}{2\lambda_{\bar{u}}(x_0)}.
\]
A straightforward computation implies that $\dfrac{\partial \hat{g}_0}{\partial \lambda_{\bar{u}}(x_0)}$ is non-positive [non- negative, respectively] whenever $x_0$ is positive [negative, respectively]. Finally, recalling that $\lambda_{\bar{u}}(x_0)= x_0 I_{\bar{u}}(x_0)$, the previous argument shows that the certainty equivalent $\hat{g}_0$ decreases as $I_{\bar{u}}(x_0)$ increases, justifying the interpretation of $I_{\bar{u}}(x_0)$ as a measure of risk aversion. Thus, positive (negative, respectively) values for $I_{\bar{u}}(x_0)$ indicate that agents who own an s-plus $x_0$ are risk averters (risk lovers, respectively). The idea of measuring risk aversion through $I_{\bar{u}}(x_0)$ goes back to de Finetti (1952), as mentioned at the beginning of this subsection, and explained, e.g., in the paper by Regazzini and Spizzichino (2011) that we are using here for the presentation of the subject. Since $\lambda_{\bar{u}}(x_0)$ is dimensionless, while $I_{\bar{u}}(x_0)$ has dimension (money)${}^{-1}$, it is sometimes useful, and it will be done in the present paper, to adopt the former as an index for risk aversion, rather than the latter, and to name $\lambda_{\bar{u}}(x_0)$ \textit{relative risk aversion index} (relative r.a.i., for short). See de Finetti (1952), Pratt (1964) and Arrow (1965).
\newline
\newline

\subsubsection{Description of the assumptions}\label{sec:SUBassumptions}
As for the description of the main assumptions, one starts by supposing there is a distinguished social welfare function derived, for example, from a common consensus, or from an authority democratically elected. Then, the aforesaid assumptions can be partitioned into three groups: The first are relative to the exchange of s-plus. The second regard the formation of the welfare function. The third concern the dynamics of the joint distribution of s-plus in a population of $N$ agents.

As far as the first group is concerned, exchanges between two agents are described through the quantification of the variation occurring in each agent's s-plus because of the exchange, assuming that one can split this variation into two components: the former resulting from the investment of the s-plus initially owned by her/him, the latter thought of as a function of the initial s-plus of the counterpart. Moreover, one supposes that each agent is: first, cautious enough to exclude contracts in which the absolute value of a possible negative variation due to the investment of an initial surplus would be greater than the aforesaid surplus, second, not so able to transform an initial subplus into a profit without the intervention of a cooperative counterpart, third, sufficiently altruist to possibly subsidize a counterpart having an initial subplus. This is tantamount to stating that each of the two components of any final s-plus is either zero or has the same sign of the corresponding generating initial s-plus. For the sake of generality, it will not be imposed that the advantage (disadvantage) of one of the two contracting parties has to correspond to a disadvantage (advantage) of the same amount for the other one. Translating these considerations into symbols, one denotes the initial s-plus of the two contracting parties, say I and II, by $v$ and $w$, respectively, and the corresponding final s-plus by $v'$ and $w'$. Then, $v'=\AI+\BI$, $w'=\AII+\BII$, where: $\AI$ ($\AII$, respectively) is the s-plus I (II, respectively) extracts from her/his own initial s-plus $v$, and $\BI$ ($\BII$, respectively) is the s-plus I (II, respectively) extracts from II (I, respectively). Now, in order to express a sort of dependence of the $A$'s and $B$'s on the initial s-plus, according to the aforesaid initial aims, one can resort to the representation $\AI=\lI v$, $\AII=\lII w$, $\BI=\rI w$, $\BII=\rII v$,  by introducing suitable coefficients $l$'s and $r$'s in an obvious manner, and get
\begin{equation}\label{interazioni}
	\begin{aligned}
	&v'= \lI v+ \rI w\\
	&w'= \rII v+ \lII w.
	\end{aligned}
\end{equation}
The coefficients $l,r$'s must be non-negative in view of the assumptions made above. In any population of contracting agents, at a given time, there will be a variety of pairs $(l,r)$. It is useful to organize all the observed pairs $(l,r)$ in the form of a statistical distribution  (of frequencies) $\tau$ to be used, if necessary, for evaluating the probability of observing a pair $(l,r)$ satisfying any given condition of interest. Thus, $\tau$ can be seen as a probability law assessed by resorting to a \textit{feasible} statistical survey. As a matter of fact, in the model we are about to propose, with reference to any observable, but not yet observed, contract, each encounter between agents will be characterized by means of the non-negative pairs $(\lI,\rI)$, $(\lII,\rII)$, thought of as realizations of two random vectors $(\tlI,\trI)$, $(\tlII,\trII)$, with $\tau$ as common p.d.. A realistic way of thinking of $\tau$ is that it is presentable as a mixture of frequency distributions, all of which is associated with a specific type of interactions, weighted with the frequency of each of these types within the population under study. It is worth recalling that \eqref{interazioni} is reminiscent of kinetic modelling for wealth distribution, like, for example, in Dr\v{a}gulescu and Yakovenko (2000), in Chapter 5 of Pareschi and Toscani (2014), and in Ajmone Marsan et al. (2016) for modelling of the same type with applications to a more extensive field of phenomena. For more specific literature, see: Angle (1986, 2006) for the connection between the surplus theory and the inequality process; Chakraborti and Chakrabarti (2000) where the concept of ``saving propensity'' is introduced; Cordier et al. (2005) for a model of economy involving both exchanges between agents and speculative trading. 

In comparison with this literature, the novelty of the present work lies in the proof of a strict relationship between attitude towards risk and inequality in the distribution of individual s-plus. Such a relationship, as explained in Section \ref{sec:eco}, enlightens plans of actions in order to reduce inequality, simultaneously with a tendency for an economic system to high-level performances.

The following are a few illustrative examples of the use of scheme \eqref{interazioni}.

\begin{example}\label{ex1}
{\rm{\small This example explains how to apply the general scheme \eqref{interazioni} to the case of public services. Let I indicate the user of a given public service and II its supplier. Here one assumes that II receives from the State an endowment lower than the real cost ($w<0$), so that one can think of the fare to be paid by I as sum of two amounts: the former, say $f_1>0$, proportioned to $w$, i.e. $f_1=f_1(w)>0$; the latter, denoted by $f_2=f_2(v)$ could take also negative values when $v$ is negative and the policy of II is that of favouring users under the poverty line. Whence, I's s-plus passes from $v$ to $v'=v-f_1(w)-f_2(v)$, that can be written according to scheme \eqref{interazioni} with $\AI=v-f_2(v)$ and $\BI=-f_1(w)$ provided that $v-f_2(v)$ has the same sign of $v$, for every $v$. On the other side, II's s-plus $w$ passes to $w'=\AII+\BII$ with $\AII=w+f_1(w)$ and $\BII=f_2(v)+g_1(v)\ind{\{v<0\}}$, where $w+f_1(w)<0$ whenever $w<0$ and $g_1(v)\geq0$ for every $v<0$ stands for an extra endowment II receives for each user under the poverty line, whenever II puts the aforesaid policy in practice. Representation \eqref{interazioni} is obtained by putting $\lI= (v-f_2(v))/v$, $\lII= (w+f_1(w))/w$, $\rI=-f_1(w)/ w$, $\rII= (f_2(v)+g_1(v)\ind{\{v<0\}})/v$.
}}
\end{example}

\begin{example}\label{ex3}
{\rm{\small This example is concerned with a market of consumer goods, in which a buyer I with an initial surplus $v$ enters into a contract with a seller II. The price of a good can be thought of as the sum of two amounts: a base price $q$ (understood as the minimum of the prices II considers as admissible) and an additional quantity $f_1$ depending both on the ability of II to negotiate and on I's initial surplus. In order to schematize this distinction according to \eqref{interazioni}, one considers $f_1$ as dependent only on $v$. Hence, $v$ passes to $v'=v-(q+f_1(v))$ that is $\BI=0$ and $\AI=v-f_1(v)-q$, where the RHS is assumed to have the same sign of $v$. As for the seller, her/his s-plus $w$ changes into $w'=w+f_1(v)+f_2(w)$ where $f_2(w)$ stands for the gain II realizes from the difference between $q$ and the cost she/he has to pay to supply the good. Thus, $\BII=f_1(v)$ and $\AII=w+f_2(w)$, where the RHS is assumed to have the same sign of $w$.
}}
\end{example}

\begin{example}\label{ex2}
{\rm{\small In this example, let I be the holder of bonds which give her/him the right to receive periodic interests, for $n$ periods, at a fixed interest rate. If $V$ stands for the present value of the periodic interests, $v$ passes to $v'=\AI+\BI$ where $\AI=v+V$ and $\BI$ might be strictly positive only for special types of bonds, which, for instance, might pay a premium related to the surplus $w$ of the issuer II: $\BI=\rI w$. On the other side, $w$ passes to $w'=w-V_1$ where $V_1$ is given by the sum of $V$ and the possible premium, provided that $\AII=w-V_1=\lII w$ has the same sign of $w$.
}}
\end{example}

\vskip 0.5cm
The second group of hypotheses concern the welfare function involved throughout the rest of the paper. 
In econophysical literature, an explicit use of utility functions is made, for example, in Toscani et al. (2013). First of all, one assumes that \textit{the joint welfare function of two agents}, I and II as usual, \textit{is strongly additive} in every contract and at every time. Whence, the joint welfare function in a contract is the sum of the individual utility functions of each of the two agents. Moreover, here it is assumed that \textit{the individual utility function}, say $\CU$, \textit{is the same for every agent}. Since in every contract made in accordance with \eqref{interazioni} each agent's s-plus changes into the sum of two gains (given by $(A_i,B_i)$ for agent $i$, $i={\rm I,II}$) drawn from two different sources, respectively $-$ both expressed in the same money of account $-$ the function $\CU$ is thought of as a real-valued function on $\R^2$. Here, the above different sources are viewed as \textit{independent}, that is, according to a definition due to de Finetti (1952a, 1952b), the two averages
\[
\dfrac{1}{2}\Big(\CU(A,B)+\CU(A+a,B+b)\Big)\quad\text{and}\quad\dfrac{1}{2}\Big(\CU(A,B+b)+\CU(A+a,B)\Big)
\]
are supposed to be equal for every $(A,B)$ and $(a,b)$. This, in its turn, implies that there are functions $u_1$, $u_2$ such that 
\[
\CU(A,B)=u_1(A)+u_2(B).
\]

At this stage, by a further assumption, one supposes that the functions $u_1$ and $u_2$ are identically equal to a distinguished utility function of money $\bar{u}$, i.e. $u_1\equiv u_2 \equiv \bar{u}$. Assuming that \textit{the relative r.a.i.} (recall Subsection \ref{sec:subsecRRAI} and, in particular, \eqref{lambdarai}) \textit{is equal to a constant $\lambda< 1/2$  for every agent}, one gets 
\[
\bar{u}(x)=|x|^{1-2\lambda}{\rm sign}(x)\qquad(x\in\R)
\]
provided that $\bar{u}(0)=0$ and $\bar{u}(1)=1$. This last proviso is not restrictive, since affine transformations of any utility function of money do not alter conclusions deducible from the use of the corresponding utility index. At this stage, the collective (positive or negative) contribution of each contract is measured by the expectation of the increment $\Delta$ of the above-defined joint welfare function, that is
\[
\begin{split}
\E(\Delta)&= \int_{[0,+\infty)^4}\Big[\Big((\CU(\lI v,\rI w)-\CU(v,0)\Big)\\
&\qquad\qquad\qquad+\Big(\CU(\lII w,\rII v)-\CU(w,0)\Big)\Big]\CT(d\lI d\rI d\lII d\rII)
\end{split}
\]
where $\CT$ is a specific joint p.d. for $(\tilde{\lI},\tilde{\rI},\tilde{\lII},\tilde{\rII})$, consistent with the assumption that $(\tilde{\lI},\tilde{\rI})$ and $(\tilde{\lII},\tilde{\rII})$ are identically distributed with common p.d. $\tau$. Hence, a straightforward computation yields 
\[
\E(\Delta)=\Big(|v|^{1-2\lambda}{\rm sign}(v)+|w|^{1-2\lambda}{\rm sign}(w)\Big)\int_{[0,+\infty)^2}(l^{1-2\lambda}+r^{1-2\lambda}-1)\tau(dldr). 
\]

Now, to complete the picture it is worth lingering over the classes of the admissible p.d.'s $\tau$ and of the admissible values for the relative r.a.i. $\lambda$, respectively. With a view to this subject, it is useful to notice that the function
\[
[0,+\infty)\ni p\mapsto \CS(p):=\int_{[0,+\infty)^2}(l^{p}+r^{p})\tau(dldr)-1\qquad\text{(read $0^0:=1$)}
\]
is convex. Without real loss of generality, one assumes that $\tau$ is \textit{non-degenerate}, so that $\CS$ turns out to be \textit{strictly} convex. This leads to distinguish two cases, i.e., $\CS(p)\geq 0$ for every $p$, or $\CS(p)<0$ for some $p$. In the former, the expectation of the increment of the joint welfare function, that is
\[
\Big(|v|^{1-2\lambda}{\rm sign}(v)+|w|^{1-2\lambda}{\rm sign}(w)\Big)\CS(1-2\lambda),
\]
is positive or negative depending on whether the initial joint utility is positive or negative, but \textit{independently} of the variation of $(1-2\lambda)$. In the latter, $\CS(p)$ admits one or two distinct zeros, and any sufficiently small deviation of $(1-2\lambda)$ from each of them causes an increase or a decrease in the initial joint utility depending on the direction of the deviation itself. Since the former is tantamount to admitting there is no feasible action to avoid rich getting richer and poor getting poorer, it is sensible to assume that an economy open to altruism refuses as inadmissible every $\tau$ leading to the first of the two cases just described. As an example, look at the situation illustrated in Example \ref{ex1}, where even the encounter between two agents with subplus might produce positive increments of their joint utility. Whence, from now on, $\tau$ will be any p.d. allowing $\CS(p)$ to be strictly negative for some $p$, and $(1-2\lambda)$ to be a zero of $\CS$.
In the presence of two distinct zeros for $\CS$, there is  a further reason to assume that the common agents relative r.a.i. is just the one corresponding to the smallest of them, since adopting the relative r.a.i. connected with the greatest zero would imply that the unit mass\footnote{Recall that a p.m. $m$ is said to be the unit (or point) mass at some $x_0$ if $m\{x_0\}=1$. The unit mass at any $x_0$ is denoted by $\delta_{x_0}$ throughout the rest of the paper.} at zero is the \textit{only} admissible long-time p.d., and suffice it to notice that such a unit mass would correspond to the unrealistic perfectly egalitarian distribution for s-plus. The technical explanation of this statement is deferred to Subsection \ref{sec:comments}, but a clue as to this fact can be drawn from a simple comparison of the situations corresponding to the two different zeros. It is interesting to notice that, passing from the smallest zero to the greatest one is accompanied by a decrease in the relative r.a.i., which amounts to a decrease (increase, respectively) in the risk aversion of surplus (subplus, respectively) holders. For this last fact, recall the relationship $\lambda_{\bar{u}}(x)= x I_{\bar{u}}(x)$ in Subsection \ref{sec:subsecRRAI}. In view of this equality, it is worth noticing that, on the one hand, (positive) values of $\lambda$ close to $1/2$ are associated with a strong risk aversion (strong risk propensity, respectively) for those agents who own a surplus (subplus, respectively). On the other hand, for negative $\lambda$'s whose absolute value becomes bigger and bigger, holders of surplus (subplus, respectively) are more and more risk lovers (risk averters, respectively).\newline
\newline

As already announced, the third and last group of hypotheses concern the dynamics $-$ from a probabilistic viewpoint $-$ of the s-plus of the $N$ agents constituting the economy of interest. Assume that the $N$ s-plus form, as time evolves, a pure jump Markov process with state space $\R^N$, or, more precisely, a Markov process with paths in $D([0,+\infty);\R^N)$, the Skorokhod space of cadlag functions taking values in $\R^N$. Moreover, suppose that the \textit{initial} joint p.d. of the $N$ s-plus makes them independent and identically distributed (i.i.d., for short) and that the generator is defined on the bounded and continuous functions $\vphi\colon\R^N\to\R$ by
\[
\begin{split}
\vphi\mapsto &\dfrac{1}{N-1}\sum_{1\leq i\neq j\leq N}\dfrac{1}{2}\int_{\R^2}\Big[\vphi(v_1,\dots,v_i+h,\dots,v_j+k,v_N)\\
&\qquad\qquad\qquad\qquad\qquad\qquad-\vphi(v_1,\dots,v_N)\Big]\tilde{\eta}(dh dk; v_i,v_j)
\end{split}
\]
where $\tilde{\eta}$ is the measure on $\R^2$ defined by
\[
\begin{split}
\tilde{\eta}(A\times B;v_i,v_j)&=\CT\Big(\{(l_i,r_i,l_j,r_j)\in\R^4:\; (l_i-1)v_i+r_i v_j\in A,\\
&\qquad\qquad\qquad\qquad\qquad\qquad\qquad(l_j-1)v_j+r_j v_i\in B\}\Big)
\end{split}
\]
for every Borel subsets $A$ and $B$ of $\R$. Roughly speaking, this generator could be derived from the following hypotheses: (i) agents interact two by two, (ii) the evolution of every s-plus is driven by the same probability law, (iii) transition probabilities depend on the form \eqref{interazioni} of the contract determining the jump. The conjunction of all these elements produces the joint p.d. $\CL_{N,t}$ of the $N$ s-plus at each time $t>0$ and, in particular, the marginal distribution of each s-plus, which is the law of actual interest to the present study. Here, this marginal p.d. is investigated by assuming that the number $N$ of agents goes to $+\infty$. In view of Theorem 3.1 in Graham and Méléard (1997), one has: \textit{For every $t$, all the one-dimensional marginal laws of $\CL_{N,t}$ are identical, say $\CL^{(1)}_{N,t}$, and there exists a probability measure {\rm(}p.m., for short{\rm)} $\mu_t$ on $\R$ such that 
\[
\sup\Big\{\Big|\CL^{(1)}_{N,t}(A)-\mu_t(A)\Big|:\;\text{$A$ is any Borel subset of $\R$}\Big\}\leq 6\dfrac{e^T-1}{N-1}
\]
for fixed $T>0$ and $t< T$}. As a consequence of this, $\mu_t$ can be viewed as s-plus p.d. of any randomly chosen agent, at time $t$, in an economy of infinitely many agents ($N\to+\infty$). With a view to the study of this p.d., it is important to recall that it turns out to be the unique solution to a distinguished Cauchy problem, for $t$ in $[0,T)$.  Assuming that the number of agents $N$ is a function of $T$, say $N=N(T)$, increasing to $+\infty$ as $T\to+\infty$ (reminiscent of the fact that time can be measured in encounters between agents) in such a way that $e^T/N(T)\to 0$, the aforesaid Cauchy problem can be extended to every $t>0$, yielding
\begin{itemize}
\item[] \textit{The limiting {\rm(}as $N\to+\infty${\rm)} s-plus law $\mu_t$ satisfies the Cauchy problem}
\begin{equation}\label{C}
\left\{ \begin{aligned}
		&\dfrac{\partial}{\partial t}\int_{\R}\psi(v)\mu_t(dv)=\int_{\R}\psi(v)Q^+(\mu_t)(dv)-\int_{\R}\psi(v)\mu_t(dv)\\
		& \qquad\qquad\qquad\qquad\qquad\qquad\qquad\qquad\qquad(t\geq 0,\;\psi\in C_b(\R;\R))\\
		& \mu_{0+}=\mu_0
		\end{aligned}
\right.
\end{equation}
\textit{where $\mu_0$ is the initial p.d. and $Q^+(\mu_t)$ the p.m. satisfying
\[
\int_{\R}\psi(v)Q^+(\mu_t)(dv):=\int_{\R^2}\int_{\R^2}\psi(lv+rw)\mu_t(dv)\mu_t(dw)\tau(dl dr)
\]
for every $\psi$ in the class $C_b(\R;\R)$ of the bounded and continuous functions from $\R$ into $\R$.}
\end{itemize}

This equation has been studied extensively in Bassetti et al. (2011) followed by Bassetti and Ladelli (2012), Bassetti and Perversi (2013) and Perversi and Regazzini (2015), especially w.r.t. the asymptotic behaviour of $\mu_t$, as $t$ goes to infinity.

At this stage it is worth condensing the whole previous reasoning into the following formal hypotheses:

\begin{itemize}
\item[($H_1$)] \textit{For any couple of agents, say I and II, the increment of the joint welfare function due to the interaction is}
\[ 
|v|^{1-2\lambda}\Big(\lI^{1-2\lambda}+\rII^{1-2\lambda}-1\Big){\rm sign}(v)+|w|^{1-2\lambda}\Big(\lII^{1-2\lambda}+\rI^{1-2\lambda}-1\Big){\rm sign}(w)
\]
\textit{for every initial s-plus $v$ and $w$ of I and II, respectively.}
\end{itemize}

\begin{itemize}
\item[($H_2$)] \textit{$\tau$ is any p.d. such that $\CS$ changes sign in $(0,+\infty)$, and the relative r.a.i. $\lambda$ is the number for which $1-2\lambda$ coincides with the smallest root of $\CS(p)=0$.}
\end{itemize}

\begin{itemize}
\item[($H_3$)] \textit{In the rest of the paper, one considers the solution $\mu_t$ of \eqref{C} as p.d. of the s-plus of an agent randomly drawn from a  population of a great number of agents.}
\end{itemize}

\begin{remark}
{\rm
It is worth mentioning that a p.d. $\tau$ with the characteristics expressed in $(H_2)$ always exists, for example any $\tau$ with topological support included in $[0,1]^2$. Such a $\tau$ is consistent with a typology of contracts in which, according to \eqref{interazioni}, the absolute value of the s-plus $A_i$ [$B_i$, respectively], $i=1,2$, drawn from the investment of agent $i$'s s-plus [drawn from the counterpart's s-plus, respectively] is smaller than the absolute value of the initial agent $i$'s [counterpart's, respectively] s-plus. On the contrary, if the support of $\tau$ contains also values greater than $1$, it is possible to find conditions under which the validity of $(H_2)$ is assured. For example, one can argue in the following terms. On the one hand, a positive relative r.a.i. $\lambda$ is consistent with the assumption that $m_l:=\int_\R l\tau(dldr)$ and $m_r:=\int_\R r\tau(dldr)$ satisfy $m_l^p+m_r^p\leq1$ for some $p$ in $(0,1]$, which entails $1-2\lambda\leq p$. On the other hand, a negative relative r.a.i. is consistent with the assumption that $m_l+m_r>1$ and
\[
(\sigma^2_l+m_l^2)^{p/2}+(\sigma^2_r+m_r^2)^{p/2}\leq 1
\]
where $\sigma_l^2:=Var(\tl)$, $\sigma_r^2:=Var(\tr)$ and $p$ is some number in $(1,2]$. In fact, in any framework of this kind, it is easy to check that $1<1-2\lambda\leq p$ holds.
}
\end{remark}

\begin{remark}\label{rem2}
{\rm
It is important to stress the fundamental role played by the shape of $\tau$ in determining the value of the r.a.i. and, as we shall see in a while, in determining the level of inequality in the (stationary) distribution of s-plus. Finally, it is worth reaffirming that the shape of $\tau$ can be modified by actions of economic policy either regulating or influencing the market ratios $(l,r)$.
}
\end{remark}

\subsection{Main results: An informal presentation}\label{sec:main}
The main results of the present study pertain to the characterization of the steady states for the s-plus distribution and to the analysis of the possible existing connections between concentration of steady states and the values of the relative r.a.i. $\lambda$. It is worth recalling that a p.m. $\mu$ is said to be a \textit{steady} (equilibrium) \textit{state} for \eqref{C} if, taking $\mu_0=\mu$ entails $\mu_t=\mu$ for every $t>0$. Since, in the present framework, the class of steady states turns out to be the same as the one of (weak) limits of solutions of \eqref{C}, as $t\to+\infty$ (see Fact \ref{propstazionarie}), the main results of the present paper can be briefly summarized as follows:

\begin{itemize}
\item[(a)] In order to reach a steady state for the distribution of the s-plus, it is necessary that the initial p.d. $\mu_0$ be weak Pareto (see next Section \ref{sec:pareto} for some remarks on weak Pareto laws) of exponent $(1-2\lambda)$, where $\lambda$ is the value of the relative r.a.i. for every agent, $-1/2<\lambda<1/2$, and that the mean of $\mu_0$ is equal to zero, whenever it is finite (see next Fact \ref{fattoCN} in Section \ref{sec:2}). The mean is obviously finite whenever $1-2\lambda>1$, i.e., in the case corresponding to common empirical evidence on the tails of observed income distributions. It is worth stressing that, in view of the relationship $\lambda_{\bar{u}}(x)=x I_{\bar{u}}(x)$, inequality $1-2\lambda>1$ is satisfied if and only if holders of surplus are risk lovers, on the contrary of holders of subplus, who turn out to be risk averters.
\item[(b)] For every initial datum like in the previous point, the long-time s-plus p.d. $\mu_\infty$ (i.e., the weak limit of $\mu_t$, as $t\to+\infty$) is weak Pareto preserving the exponent $(1-2\lambda)$. See Facts \ref{fattoCS},\ref{fattocode} in Section \ref{sec:2}. 
Moreover, if $0\leq \lambda<1/2$, then $\mu_\infty$ is maximally concentrated, whilst, if $-1/2<\lambda<0$, the measure of inequality in $\mu_\infty$ decreases as $(1-2\lambda)$ increases. Apropos of this, see Fact \ref{fattoConc1} and Proposition \ref{PropConcwP} Section \ref{sec:pareto}.
\item[(c)] According to next Fact \ref{explosion} in Section \ref{sec:2}, if the relative r.a.i. $\lambda$ belongs to $(-1/2,0)$ and the weak Pareto initial p.d. has exponent $\a$ in $(0,2)$, then $\mu_t$ tends towards the egalitarian s-plus distribution (at zero) whenever $\a>1-2\lambda$, whilst, for values of $\a$ smaller than $1-2\lambda$, the unit probability mass distributed according to $\mu_t$, at time $t$, escapes to $\infty$, as $t\to+\infty$.
\end{itemize}

The rest of the paper is organized as follows. Section \ref{sec:pareto} recalls the definition of weak Pareto law and deals with the problem of how to measure its concentration according to the notion of Gini's index given in Subsection \ref{sec:subsecConc}. The subsequent Section \ref{sec:2} contains a precise formulation of the above results. Section \ref{sec:eco} deals with a few implications of those results, with a view to their fitness for political action. Section \ref{sec:complementi} concludes the paper with technical details concerning the proofs of propositions stated in Sections \ref{sec:pareto} and \ref{sec:2}.


\section{Some remarks on weak Pareto laws}\label{sec:pareto}
Following Mandelbrot (1960), a p.m. $\pi_\a$ on $\R$ is said to be a \textit{weak Pareto} law of exponent $\a>0$ if 
\begin{equation}\label{limiti}
\lim_{x\to+\infty}x^{\a}\pi_{\a}(-\infty,-x]=c_1,\qquad\lim_{x\to+\infty}x^{\a}\pi_{\a}(x,+\infty)=c_2
\end{equation}
with $c_1+c_2>0$. This section aims at supplying an interpretation of $\a$ as a measure of concentration w.r.t. the perfectly egalitarian distribution, according to the remarks made in Subsection \ref{sec:subsecConc}. First, notice that, in the subclass of $\CF_0$ containing the \textit{strict Pareto} p.d. functions $F$ defined by
\begin{equation}\label{Pareto}
dF(x)=\ind{[x_0,+\infty)}(x)\dfrac{\a x_0^\a}{x^{\a+1}}dx \qquad(x>x_0)
\end{equation}
with $x_0>0$ and $\a>1$, the partial ordering induced by the Lorenz curve becomes total and is simply directed by the exponent $\a$: The concentration curve of \eqref{Pareto} reads
\[
\vphi_{F}(\theta)= 1-(1-\theta)^{1-1/\a}\qquad(\theta\in (0,1))
\]
so that it is plain to check that it increases, i.e. economic inequality decreases, as $\a$ increases. By the way, Benini (1908) was the first to notice this fact. Something similar to this holds true for the entire class of weak Pareto laws. In fact, using p.d. functions and their corresponding p.m.'s symbols interchangeably, one has 

\begin{proposition}\label{PropConcwP}
Let $\pi_{\a_1}$ and $\pi_{\a_2}$ be p.m.'s on $\R$ such that
\[
\lim_{x\to+\infty}x^{\a_i}\pi_{\a_i}(-\infty,-x]=c^{-}_{i},\qquad\lim_{x\to+\infty}x^{\a_i}\pi_{\a_i}(x,+\infty)=c^{+}_{i}
\]
with $c^{-}_{i}+c^{+}_{i}>0$ for $i=1,2$, and $1<\a_1<\a_2<2$. Then, there is a suitable $\overline{\theta}$ in $(0,1)$ such that 
\[
\vphi_{\pi_{\a_1}}(\theta)\leq \vphi_{\pi_{\a_2}}(\theta)
\]
holds true for every $\theta$ in $(\overline{\theta},1)$.
\end{proposition} 

The proof is deferred to Subsection \ref{sec:conc}. 

The last proposition can be used to justify the adoption of any strictly decreasing function of $\a$, taking values in $[0,1]$, as a measure of the concentration within the family of the weak Pareto laws with $\a>1$. In fact, apart from the connection with the behaviour of the function $\vphi_F$ associated with \eqref{Pareto}, one ought to recall Herzel (1968), where a few concentration indices are defined on the basis of the slope of $\vphi_F$ in small neighbourhoods of $0$ and $1$, respectively. For $\a$ in $(0,1]$, the following argument is conducive to consider the corresponding weak Pareto law as maximally concentrated. Indeed, let $F_\a$ be the p.d. function associated with $\pi_\a$ and $F_\a^{(\o)}$ be the conditional restriction of $F_\a$ to $(-\o,\o)$, for any $\o$ for which $(-\o,\o)$ is charged by $F_\a$. Define $A_{F_\a}^{(\o)}$ to be the corresponding p.d. function of absolute values, that is,
\[
\R\ni x\mapsto A_{F_\a}^{(\o)}(x)=\dfrac{A_{F_\a}(x\wedge \o)}{A_{F_\a}(\o)}.
\]
As to the corresponding concentration function $\vphi_{F_{\a}^{(\o)}}$, in Subsection \ref{sec:conc} the following fact will be proved.

\begin{fact}\label{fattoConc1}
If $\a$ belongs to $(0,1]$, then
\[
\lim_{\o\to+\infty}\vphi_{F_{\a}^{(\o)}}(\theta)=0
\]
for every $\theta$ in $(0,1)$.
\end{fact}

Hence, looking at $\vphi_{F_{\a}^{(\o)}}$ as an approximation of $\vphi_{F_{\a}}$, for large values of $\o$ $-$ as a consequence of the fact that $A_{F_\a}^{(\o)}$, in its turn, is an approximation of $A_{F_\a}$ $-$ one can consider the weak Pareto p.d.'s $\pi_\a$ with $0<\a\leq 1$ as \textit{maximally concentrated} from Gini's index viewpoint.


In view of the aims of the present paper, it is not necessary to define any specific concentration index, since the main conclusions will depend on one of the following three relations: $\a=1-2\lambda$, $\a>1-2\lambda$, $\a<1-2\lambda$, $\a$ being the exponent of a weak Pareto law describing the initial datum $\mu_0$ in \eqref{C}. For reasons that will be explained very soon $-$ but that can be easily deduced from Subsection \ref{sec:main} $-$ an economic situation is called \textit{conservation-oriented} or \textit{egalitarianism-oriented} or \textit{inequality-oriented}, depending on which of the above three relations is satisfied.

\section{Precise formulation of the main results}\label{sec:2}
This section is devoted to the complete formulation of the main results already expounded in Subsection \ref{sec:main}. In point of fact, our role is limited to rearranging, in a form suitable for the present setting, a few well-known propositions proved in Bassetti et al. (2011), Bassetti and Perversi (2013) and Perversi and Regazzini (2015). It is worth recalling that the limit $\mu_\infty$ of the solution $\mu_t$ of \eqref{C}, as $t\to+\infty$, is meant in the sense of weak convergence of p.m.'s: $\mu_t$ converges \textit{weakly} to the p.m. $\mu_\infty$ on $\R$ if $\mu_t(-\infty,x]\to\mu_\infty(-\infty,x]$ at each $x$ such that $\mu_\infty\{x\}=0$. It is useful to notice the following preliminary statement, to be proved in Subsection \ref{sec:proofs}.

\begin{fact}\label{propstazionarie}
If a p.d. is a steady state for \eqref{C}, then there exists an initial p.d. for which the corresponding solution $\mu_t$ of \eqref{C} converges weakly, as $t\to+\infty$, to the above steady state. Conversely, if, for a given initial datum, $\mu_t$ converges weakly to a p.d., then such a limiting p.d. is a steady state.
\end{fact}

Therefore, characterizing equilibria is equivalent to characterizing the limits of solutions of \eqref{C} as initial data $\mu_0$ vary. As to the existence of steady states, there is an economic literature in which the existence of steady states is guaranteed independently of the form of the initial data but no analytical definition of those states is supplied. On the contrary, there are models allowing perfectly specified equilibria sometimes obtained in the presence of particular initial p.d.'s. 
Coming back to the framework stated in Subsection \ref{sec:SUBassumptions}, one is faced with an intermediate situation, that is, convergence of $\mu_t$ occurs only in the presence of a distinguished class of initial data $\mu_0$ $-$ essentially, the class of the weak Pareto laws $-$ and the steady states are completely determined in the form of their Fourier-Stieltjes transforms. These statements, derived from Theorems 1 and 4 in Perversi and Regazzini (2015), are condensed into the following

\begin{fact}\label{fattoCN}
Assume $(H_1)$-$(H_2)$ are in force together with the extra conditions that 
\begin{equation}\label{conttau}
\text{the marginal p.d.'s of $\tau$ are continuous}
\end{equation}
and that 
\begin{equation}\label{supporto}
\begin{split}
& \text{every open disk centred at any point $(x,y)\in[0,+\infty)^2$,}\\
& \text{for which $x^{1-2\lambda}+y^{1-2\lambda}=1$, has strictly positive $\tau$-probability.}
\end{split}
\end{equation}
Then, in order that the solution of \eqref{C} converge, as $t\to+\infty$, to a steady state, it is necessary that one of the following conditions hold true:
\begin{itemize}
\item[{\rm(a)}] The relative r.a.i. $\lambda$ is smaller than $-1/2$ and, at time zero, the total income reduces to the subsistence level {\rm(}that is, $\mu_0$ is the point mass at $0${\rm)}.
\item[{\rm(b)}] The relative r.a.i. $\lambda$ is equal to $-1/2$ and, at time zero, the distribution of s-plus has zero mean and finite variance $\sigma_0^2$.
\item[{\rm(c)}] The relative r.a.i. $\lambda$  belongs to $(0,1/2)$ and the distribution at time zero is weak Pareto {\rm(}according to \eqref{limiti}{\rm)} with exponent $\a=1-2\lambda$.
\item[{\rm(d)}] The relative r.a.i. $\lambda$ is equal to $1/2$ and the distribution at time zero is weak Pareto with exponent $\a=1$ and $c_1=c_2$ in \eqref{limiti}.
\item[{\rm(e)}] The relative r.a.i. $\lambda$  belongs to $(-1/2,0)$ and the distribution at time zero is weak Pareto {\rm(}according to \eqref{limiti}{\rm)} with zero mean and exponent $\a=1-2\lambda$. 
\end{itemize}
\end{fact}

\begin{remark}
{\rm As far as the extra conditions \eqref{conttau} and \eqref{supporto} are concerned, they play a mere technical role, that can be noted in the proofs of Theorems 1, 2, 4 and 5 in Perversi and Regazzini (2015). In any case, the meaning of \eqref{conttau} is clear, while \eqref{supporto} is automatically satisfied whenever the increment $\Delta$ (and not only $\E(\Delta)$) of the joint welfare function is equal to zero in every contract.
}
\end{remark}
\begin{remark}
{\rm It is noteworthy that, according to Fact \ref{fattoCN}, if the expectation of $\mu_0$ is finite, then it must be equal to zero. This is reminiscent of a status of things in which surplus balance subplus, in mean. This comment can be improved by resorting to the argument, used in Subsection 3.3 of Perversi and Regazzini (2015), entailing that, as $t\to+\infty$, the law $\mu_t$ tends to concentrate on $(a,+\infty)$ [$(-\infty,-a)$, respectively], for every $a>0$, whenever the aforesaid expectation is strictly positive [negative, respectively]. 
}
\end{remark}

At this stage, one can state the form of the limiting steady states, which turn out to be scale-mixtures of \textit{stable laws}. The appearance here of stable laws is not accidental. In fact, on the one hand, a link can be established between the central limit theorem and the convergence of the solution $\mu_t$ of problem \eqref{C} $-$ see Subsection \ref{sec:comments} $-$ and, on the other hand, it is well-known that stable laws are limiting laws of normed sums of i.i.d. random numbers. Unfortunately, closed forms for stable laws $g$ are known in only a very small number of cases, but there is an explicit expressions for their Fourier-Stieltjes transforms $\hat{g}$, that is
\begin{equation}\label{cfstabile}
\R \ni \xi \mapsto \hat{g}(\xi;\a,\chi,k_\a,\gamma)= \exp\Big\{i\chi\xi-k_\a|\xi|^\a\Big(1-i\gamma \dfrac{\xi}{|\xi|}\o(\xi,\a)\Big)\Big\},
\end{equation}
where $\a$, $\gamma$, $\chi$, $k_\a$ are constants ($k_\a\geq0$, $0<\a\leq2$, $|\gamma|\leq 1$, with the proviso that $\gamma=0$ if $\a=2$) and 
\[
\begin{aligned}
\o(\xi,\a)&= \tan (\pi\a/2) & \a\neq1\\
& =2\pi^{-1}\log|\xi| & \a=1.
\end{aligned}
\]

In what follows, $\hat{g}(\cdot;\a,\chi,k_\a,\gamma)$ is extended to any strictly positive $\a$ by putting
\[
\hat{g}(\xi;\a,\chi,k_\a,\gamma)=e^{i\chi\xi}\qquad(\xi\in\R,\;\a>2).
\]

\begin{fact}\label{fattoCS}
If $\mu_t$ is solution of \eqref{C} under $(H_1)$-$(H_2)$, and if one of the conditions {\rm (a)-(e)} holds true, then $\mu_t$ converges, as $t\to+\infty$, to a p.m. $\mu_\infty$ having Fourier-Stieltjes transform 
\[
\hat{\mu}_\infty(\xi)=\int_{0}^{+\infty}\hat{g} (\xi m^{1/(1-2\lambda)};1-2\lambda, \chi,k_{1-2\lambda} ,\gamma)\nu_{1-2\lambda}(dm)\qquad(\xi\in\R)
\]
where $\nu_{1-2\lambda}$ is a uniquely determined p.m. on $[0,+\infty)$. Moreover, the function $\hat{g}$ has one of the following forms, depending on which condition {\rm (a)-(e)} is met by $\mu_0$ in {\rm Fact \ref{fattoCN}}:
\begin{itemize}
\item[{\rm(a')}] $\hat{g}(\xi ;1-2\lambda, \chi,k_{1-2\lambda}, \gamma)\equiv 1$.
\item[{\rm(b')}] $\qquad''\qquad\equiv e^{-\sigma_0^2\xi^2/2}$.
\item[{\rm(c')}] $\qquad''\qquad\equiv \exp\Big\{-k_{1-2\lambda}|\xi|^{1-2\lambda}\Big(1-i\gamma \dfrac{\xi}{|\xi|}\tan(\pi(1-2\lambda)/2)\Big)\Big\}$.
\item[{\rm(d')}] $\qquad''\qquad\equiv \exp\Big\{i\chi\xi-k_{1}|\xi|\Big\}$.
\item[{\rm(e')}] $\qquad''\qquad\equiv \exp\Big\{-k_{1-2\lambda}|\xi|^{1-2\lambda}\Big(1-i\gamma \dfrac{\xi}{|\xi|}\tan(\pi(1-2\lambda)/2)\Big)\Big\}$.
\end{itemize}
\end{fact}

The values of $k_{1-2\lambda}$ and $\gamma$ are given in Subsection \ref{sec:comments}, while $\nu_{1-2\lambda}$ is identified therein with the limiting law of a distinguished sequence of random numbers. 

\begin{remark}
{\rm Conditions on $\tau$ under which $\mu_{\infty}$ is just a stable law (that is, $\nu_{1-2\lambda}$ reduces to a unit mass at some point) are provided by Proposition 2 in Bassetti et al. (2011). More precisely: \textit{If $-1/2\leq\lambda<1/2$, then the limiting p.m. $\mu_{\infty}$ reduces to a stable law if and only if $\tau$ concentrates the entire mass on $\{(x,y)\in[0,+\infty)^2:\; x^{1-2\lambda}+y^{1-2\lambda}=1\}$}}.
\end{remark}

On the basis of rich empirical evidence of tails of observed income distributions, the most realistic forms ought to be those which satisfy conditions (e) and (e'), under which the tails of $\mu_\infty$ are of the same type as the ones of $\mu_0$. This fact can be ascertained mathematically, even if it is impossible to obtain explicit forms for mixtures $\mu_\infty$. 

\begin{fact}\label{fattocode}
Under $(H_1)$-$(H_2)$, let $\mu_\infty$ be the same p.m. as in {\rm Fact \ref{fattoCS}} and $\nu_{1-2\lambda}$ be non-degenerate. Then:
\begin{itemize}
\item[{\rm(i)}] If one of the conditions among {\rm (c)-(e)} is in force and $c_1\cdot c_2>0$, then
\[
\lim_{x\to+\infty}x^{1-2\lambda}\mu_{\infty}((-\infty,-x])=c_1 \quad\text{and}\quad \lim_{x\to+\infty}x^{1-2\lambda}\mu_{\infty}((x,+\infty))=c_2.
\]
\item[{\rm(ii)}] If one of the conditions among {\rm (c)-(e)} is in force and $c_1>0$, $c_2=0$, then
\[
\lim_{x\to+\infty}x^{1-2\lambda}\mu_{\infty}((-\infty,-x])=c_1 \quad\text{and}\quad \mu_{\infty}((x,+\infty))=O\Big(\dfrac{1}{x^p}\Big)
\]
as $x\to+\infty$, for every positive $p$ such that $\CS(p)<0$.
\item[{\rm(iii)}] If one of the conditions among {\rm (c)-(e)} is in force and $c_1=0$, $c_2>0$, then
\[
\mu_{\infty}((-\infty,-x])=O\Big(\dfrac{1}{x^p}\Big) \quad\text{and}\quad \lim_{x\to+\infty}x^{1-2\lambda}\mu_{\infty}((x,+\infty))=c_2
\]
as $x\to+\infty$, for every positive $p$ such that $\CS(p)<0$.
\item[{\rm(iv)}] If {\rm (b)} is in force and $\CS$ admits only one zero, then
\[
\mu_{\infty}((-\infty,-x])=\mu_{\infty}((x,+\infty))=o\Big(\dfrac{1}{x^p}\Big)
\]
as $x\to+\infty$, for every positive $p$.
\item[{\rm(v)}] If {\rm (b)} is in force and $\CS$ admits two distinct zeros, say $1-2\lambda$ and $l>1-2\lambda$, then
\[
\mu_{\infty}((-\infty,-x])=\mu_{\infty}((x,+\infty))\geq \dfrac{M}{2 l^{1+\delta}}\dfrac{1}{x^l(\log x)^{1+\delta}}
\]
for every positive $\delta$, for sufficiently large $x$ and a suitable constant $M>0$.
\end{itemize}
\end{fact}

As a consequence, the economic situation corresponding to (e) turns out to be conservation-oriented. Moreover, a natural question arises, wondering about the limiting behaviour of the s-plus distribution when $\mu_0$ exhibits a higher of lower level of inequality than the one prescribed by (e), i.e. if $\a$ takes a value either smaller or greater than $(1-2\lambda)$.

\begin{fact}\label{explosion}
In an economy satisfying $(H_1)$-$(H_2)$ and \eqref{conttau}, with relative r.a.i. $\lambda$ in $(-1/2,0)$, let the initial datum be a weak Pareto law of exponent $\a$ in $(0,2)$. Then: 
\begin{itemize}
\item[{\rm(i)}] The solution $\mu_t$ to \eqref{C} converges weakly to the point mass at zero if $\a>1-2\lambda$ {\rm(egalitarianism-oriented)}.
\item[{\rm(ii)}] $\lim_{t\to+\infty}\mu_t((-a,a))=0$ for every $a>0$ if $\a<1-2\lambda$ {\rm(inequality-oriented)}. 
\end{itemize}
\end{fact}

A hint at the proof of this fact will be given in Subsection \ref{sec:proofs}. The argument used therein shows that the convergence to the point mass at zero, like in (i), and the vague convergence to the null measure, like in (ii), occur also when $\lambda$ belongs to $[0,1/2)$, depending on whether the weak Pareto initial datum has exponent $\a$ greater or smaller than $1-2\lambda$.

\section{Economic implications}\label{sec:eco}

What actually characterizes a certain number of advanced economies is that the level of inequality increases in time, so that the middle class, as a consequence of Fact \ref{explosion}, might be hollowed out. Does the analysis, developed in previous sections, suggest actions suitable for thwarting this process, generally seen as highly negative? To answer, one must take note that, from Fact \ref{explosion}(ii) combined with Fact \ref{fattoCN} the process stems from conditions in which the concentration of initial weak Pareto law is too high, w.r.t. to the given relative r.a.i., to be preserved in time. Then, aiming at reversing the trend, in view of Fact \ref{fattoCS} there are two kinds of actions that can be taken either separately or jointly. After fixing a bearable level of concentration, corresponding to a weak Pareto law of exponent $\a_0$ one ought act with redistributive programs (social expenditure to reduce poverty, highly progressive tax system, government policies such as the increased access to higher education) to change the actual s-plus distribution into a weak Pareto law of exponent $\a_0\leq 1-2\lambda$. Having done this, there are the following two alternatives facing the political authority: First, $\a_0$ is equal to $(1-2\lambda)$, $\lambda$ being the actual relative r.a.i.. Second, $\a_0$ is strictly smaller than $(1-2\lambda)$. Under the former circumstance, no further measure has to be taken. In the latter, more provisions are needed to modify the form of $\tau$ in order that, as a result, the smallest zero of $\CS$ coincide with $\a_0$. To this end, it is worth recalling the well-known inequality
\begin{equation}\label{inequalities}
\begin{split}
\int_0^{+\infty}x^p dF(x)&=p\int_0^{+\infty}x^{p-1}(1-F(x))dx\leq p\int_0^{+\infty}x^{p-1}(1-G(x))dx\\
&=\int_0^{+\infty}x^p dG(x).
\end{split}
\end{equation}
valid for any $p>0$ and any pair of p.d. functions $F$, $G$ supported by $[0,+\infty)$ and such that $F(x)\geq G(x)$ for every $x$. In fact, since one has to change the relative r.a.i. in such a way that it may increase towards $(1-\a_0)/2$, and since $p\mapsto \CS(p)$ depends on the marginal p.d.'s of $\tau$, \eqref{inequalities} indicates that one has to make sure that, as a consequence of the aforesaid provisions (e.g., to thwart any phenomenon of rent seeking), the marginals increase pointwise, i.e., to make sure that the agents content themselves with stochastically smaller $\tl$ and, especially, with stochastically smaller $\tr$. 

On the contrary, under the situation envisaged in Fact \ref{explosion}(i), the economy is egalitarianism-oriented, so that it is reasonable to expect that such a situation may encounter obstacles. A certain level of inequality could be inevitable: In fact, in any economy, individuals who work harder and longer than others will basically receive a reward for the energy they put into their work. Specularly w.r.t. the inequality-oriented situation, one can obtain an increase of inequality either by changing the actual $\mu_t$ into a weak Pareto law with exponent close to (and not less than) a desired level $\a_0$, or by modifying the marginal p.d.'s of $\tau$ in such a way that $\tl$ and $\tr$ become stochastically greater (i.e. the marginals decrease poinwise). This last modification, reminiscent of the so-called ``race to the bottom'' phenomenon, can be carried out, for example, by trying to weaken business regulations.

In conclusion, the sense of the above considerations is that they lead to single out specific actions in order to increase or diminish the inequality in the p.d. of s-plus. On the one hand, one could act directly on the incomes by means of transfers of money. On the other hand, one could act indirectly with actions influencing on the form of $\tau$ with a view to changing the attitude of the agents in front of risk.

\section{Technical complements}\label{sec:complementi}
This section gathers a number of comments and proofs necessary to complete, mainly from a technical viewpoint, some of the arguments developed in previous sections.

\subsection{On the concentration function of weak Pareto laws}\label{sec:conc}
This subsection supplies the proofs of the statements concerning the Lorenz curve of a weak Pareto law, situated in Section \ref{sec:pareto}.
\begin{proof}[Proof of {\rm Proposition \ref{PropConcwP}}]
Let $\mu_{\a_i}$ be any of the two laws considered in the wording of the proposition. Then,
\[
\lim_{x\to+\infty}x^{\a_i}\big(1-A_{\mu_{\a_i}}(x)\big)=c_i^++c_i^-
\]
and, hence, for every positive $\veps$ there is a positive $\bar{x}_i$ such that
\[
1-\dfrac{c_i^++c_i^-+\veps}{x^{\a_i}}\leq A_{\mu_{\a_i}}(x)\leq 1-\dfrac{c_i^++c_i^- -\veps}{x^{\a_i}}
\]
holds true for every $x\geq\bar{x}_i$. Thus, one can immediately obtain that the inequalities
\[
\Big(\dfrac{c_i^++c_i^- -\veps}{1-\theta}\Big)^{1/\a_i}\leq A_{\mu_{\a_i}}^{-1}(\theta)\leq \Big(\dfrac{c_i^++c_i^- +\veps}{1-\theta}\Big)^{1/\a_i}
\]
are valid for every $\theta>\bar{\theta}_i$ , with $\bar{\theta}_i:=1-(c_i^++c_i^- -\veps)/\bar{x}_i^{\a_i}$. Whence, since
\[
\vphi_{\mu_{\a_i}}(\theta)=\dfrac{1}{\CM(A_{\mu_{\a_i}})}\int_0^\theta A_{\mu_{\a_i}}^{-1}(t)dt= 1-\dfrac{1}{\CM(A_{\mu_{\a_i}})}\int_\theta ^1 A_{\mu_{\a_i}}^{-1}(t)dt,
\]
from the above inequalities one deduces that
\[
\begin{split}
&1-\dfrac{1}{\CM(A_{\mu_{\a_i}})}\dfrac{\a_i(c_i^++c_i^- +\veps)^{1/\a_i}}{\a_i-1}(1-\theta)^{1-1/\a_i} \\
&\qquad\qquad\qquad\qquad\leq  \vphi_{\mu_{\a_i}}(\theta)\leq 1-\dfrac{1}{\CM(A_{\mu_{\a_i}})}\dfrac{\a_i(c_i^++c_i^- -\veps)^{1/\a_i}}{\a_i-1}(1-\theta)^{1-1/\a_i}
\end{split}
\]
is valid for every $\theta>\bar{\theta}_i$. Then, to complete the proof, it is enough to verify that
\[
\begin{split}
&1-\dfrac{1}{\CM(A_{\mu_{\a_1}})}\dfrac{\a_1(c_1^++c_1^- -\veps)^{1/\a_1}}{\a_1-1}(1-\theta)^{1-1/\a_1}\\
&\qquad\qquad\leq 1-\dfrac{1}{\CM(A_{\mu_{\a_2}})}\dfrac{\a_2(c_2^++c_2^- +\veps)^{1/\a_2}}{\a_2-1}(1-\theta)^{1-1/\a_2}
\end{split}
\]
holds for every $\overline{\theta}<\theta<1$, with  
\[
\overline{\theta}:=\max\Big\{1-\Big[\dfrac{\CM(A_{\mu_{\a_2}})}{\CM(A_{\mu_{\a_1}})}\dfrac{(c_1^++c_1^- -\veps)^{1/\a_1}}{(c_2^++c_2^- +\veps)^{1/\a_2}}\dfrac{\a_1(\a_2-1)}{\a_2(\a_1-1)}\Big]^{\a_1\a_2/(\a_2-\a_1)},\bar{\theta}_1,\bar{\theta}_2\Big\}.
\]
\end{proof}

\begin{proof}[Proof of {\rm Fact \ref{fattoConc1}}]
The proof is obtained by showing that the concentration function $\vphi_{F_\a^{(\o)}}$ associated with the p.d. function $A^{(\o)}_{F_\a}$, introduced immediately before the statement of Fact \ref{fattoConc1}, converges pointwise to zero in $(0,1)$, as $\o$ goes to $+\infty$. First, notice that $\lim_{x\to+\infty}x^\a(1-A_{F_\a}(x))$ exists and is equal to a positive number, say $l$. Thus, for every $\veps>0$, there is a positive $\bar{x}$ such that
\begin{equation}\label{stime1}
1-\dfrac{l+\veps}{x^\a}\leq A_{F_\a}(x)\leq 1-\dfrac{l-\veps}{x^\a}
\end{equation}
holds true for every $x>\bar{x}$, and so, assuming that $\o$ is large enough to be strictly greater than $\bar{x}$, it is immediate to check that the inequalities
\begin{equation}\label{stime2}
\Big(\dfrac{l-\veps}{1-\theta A_{F_\a}(\o)}\Big)^{1/\a}\leq (A_{F_\a}^{(\o)})^{-1}(\theta)\leq \Big(\dfrac{l+\veps}{1-\theta A_{F_\a}(\o)}\Big)^{1/\a}
\end{equation}
are met at each $\theta>\bar{\theta}:=(1-(l-\veps)\bar{x}^{-\a})/A_{F_\a}(\o)$. The proof proceeds by separating the case in which $\a<1$ from the one of $\a=1$.

If $\a<1$ and $\theta>\bar{\theta}$, \eqref{stime2} entails
\[
\begin{split}
\vphi_{F_\a^{(\o)}}(\theta)&=1-\dfrac{1}{\CM(A_{F_{\a}}^{(\o)})}\int_\theta^1 (A_{F_\a}^{(\o)})^{-1}(t) dt\\
&\leq 1-\dfrac{\a(l-\veps)^{1/\a}}{\CM(A_{F_{\a}}^{(\o)})(1-\a)A_{F_\a}(\o)}\Big[\Big(1-A_{F_\a}(\o)\Big))^{1-1/\a}\\
&\qquad\qquad\qquad\qquad\qquad\qquad\qquad\qquad-\Big(1-\theta A_{F_\a}(\o)\Big))^{1-1/\a}\Big].
\end{split}
\]
At this stage, from \eqref{stime1} and straightforward computations, one can bound $\CM(A_{F_{\a}}^{(\o)})$ according to
\[
\begin{split}
\CM(A_{F_{\a}}^{(\o)})=\Big(\int_0^{\bar{x}}+\int_{\bar{x}}^\o \Big)(1-A_{F_{\a_1}}^{(\o)}(x))dx\leq K +\o^{1-\a}\dfrac{\a(l+\veps)+2\veps(1-\a)}{(1-\a)A_{F_\a}(\o)}
\end{split}
\]
where $K$ is a suitable constant. Therefore, combination of this bound with \eqref{stime1} concludes the proof, when $\a<1$, since, for every $\theta>\bar{\theta}$, 
\[
\begin{split}
\lim_{\o\to+\infty}\vphi_{F_\a^{(\o)}}(\theta)&\leq \lim_{\o\to+\infty} \Big[1-\dfrac{(l-\veps)^{1/\a}}{2\veps+\a(l+\veps)/(1-\a)}\dfrac{\a}{1-\a}\Big( (l+\veps)^{1-1/\a}\\
&\qquad\qquad\qquad\qquad\qquad\qquad-\o^{\a-1}(1-\theta A_{F_\a}(\o))^{1-1/\a}\Big)\Big]\\
&= 1-\Big(\dfrac{l-\veps}{l+\veps}\Big)^{1/\a}\dfrac{1}{1+2\veps(1-\a)/(l+\veps)}
\end{split}
\]
is positive and arbitrarily small in view of the arbitrariness of $\veps$.

If $\a=1$ and $\theta>\bar{\theta}$, one can argue as in the case $\a<1$ to obtain
\[
\begin{split}
\vphi_{F_1^{(\o)}}(\theta)&\leq 1-\dfrac{1}{\CM(A_{F_{1}}^{(\o)})}\int_\theta^1 (l-\veps)(1-tA_{F_1}(\o))^{-1} dt\\
& 1-\dfrac{l-\veps}{\CM(A_{F_{1}}^{(\o)})A_{F_1}(\o)}\Big[\log(1-\theta A_{F_1}(\o))-\log(1-A_{F_1}(\o))\Big]
\end{split}
\]
and, mimicking $-$ \textit{mutatis mutandis} $-$ the same computations to bound $\CM(A_{F_{1}}^{(\o)})$, one easily gets
\[
\lim_{\o\to+\infty}\vphi_{F_1^{(\o)}}(\theta)\leq 1-\dfrac{l-\veps}{l+\veps}
\]
completing the proof in view of the arbitrariness of $\veps$. 
\end{proof}

\vskip 1cm

Given useful information on the concentration of weak Pareto laws, we proceed with the description of a probabilistic representation of $\mu_t$ evoked in some passages in the previous sections.

\subsection{Comments on (and consequences of) a well-known probabilistic representation of the solution $\mu_t$}\label{sec:comments}
As recalled in Section 2.1 in Bassetti et al. (2011), the Fourier-Stieltjes transform of the solution $\mu_t$ to \eqref{C} can be expressed as 
\[
\hat{\mu}_t(\xi)=\sum_{n\geq1}e^{-t}(1-e^{-t})^{n-1}\hat{q}_n(\xi,\vphi_0)
\]
where $\hat{q}_1(\xi;\vphi):=\hat{\mu}_0(\xi)$ and, for every $n\geq2$, 
\begin{equation}\label{qn}
\hat{q}_n(\xi;\vphi):=\dfrac{1}{n-1}\sum_{j=1}^{n-1}\int_{[0,+\infty)^2}\hat{q}_{j}(l\xi)\hat{q}_{n-j}(r\xi)\tau(dldr).
\end{equation}
In view of an argument explained therein, and which, in its turn, draws on previous papers McKean (1966), Gabetta and Regazzini (2006), one verifies that $\mu_t$ is the p.d. of a random number
\[
V_t:=\sum_{j=1}^{\tnu_t}\tbe_{j,\tnu_t}X_j
\] 
for every $t>0$, where $\tnu_t$, $\tbe_{j,\tnu_t}$ and $X_j$ ($j=1,2,\dots$) are random elements defined on a sufficiently large measurable space $(\O,\CF)$, endowed with a suitable p.m. $\P$ according to which:
\begin{itemize}
\item $\tnu:=(\tnu_t)_{t\geq0}$ is an integer-valued stochastic process such that
\[
\P\{\tnu_t=n\}=e^{-t}(1-e^{-t})^{n-1}\qquad(n=1,2,\dots)
\]
for every $t>0$.
\item $\tX:=(X_j)_{j\geq1}$ is a sequence of i.i.d. random numbers with common p.d. $\mu_0$.
\item $\tbe:=(\tbe_{j,n}:\;j=1,\dots,n)_{n\geq1}$ is a triangular array recursively defined by
\begin{equation}\label{betaricorsivo}
\begin{array}{ll}
&\tbe_{1,1}=1\\
&(\tbe_{1,n+1},\dots,\tbe_{n+1,n+1})=(\tbe_{1,n},\dots,\tbe_{\ti_n-1,n},\tbe_{\ti_n,n}\tl_n,\tbe_{\ti_n,n}\tr_n, \\
&\qquad\qquad\qquad\qquad\qquad\qquad\qquad\tbe_{\ti_n+1,n},\dots,\tbe_{n,n})\quad(n\geq1)
\end{array}
\end{equation}
where $\ti:=(\ti_n)_{n\geq1}$ is a sequence of independent integer-valued random numbers, each $\ti_n$ being uniformly distributed on $\{1,\dots,n\}$, and $(\tl,\tr)=((\tl_n,\tr_n))_{n\geq1}$ is a sequence of i.i.d. random vectors with $\tau$ as common p.d..
\item $\{\tnu,\tX,\ti,(\tl,\tr)\}$ forms an independency w.r.t. $\P$.
\end{itemize}
Whence, $\mu_t$ turns out to be the law of a weighted sum, with random weights, of i.i.d. random numbers. This fact is conducive to studying the limiting behaviour of $\mu_t$, as $t\to+\infty$, from the point of view of the central limit theorem of probability theory. In point of fact, what it happens is that, roughly speaking, the conditional p.d. of $V_t$, given $\tbe$, converges weakly to a stable law of exponent $(1-2\lambda)$ characterized by Fourier transform
\[
\R\ni\xi\mapsto \hat{g}(\xi (M_{\infty}^{(1-2\lambda)})^{1/(1-2\lambda)};1-2\lambda, \chi,k_{1-2\lambda} ,\gamma)
\]
where $M_{\infty}^{(1-2\lambda)}$ is a random number whose p.d. is the same $\nu_{1-2\lambda}$ as in Facts \ref{fattoCS} and \ref{fattocode}. The reader is referred to Bassetti et al. (2011) for the proof of this statement. It is worth noticing that, in view of Proposition 2 therein, if $(1-2\lambda)$ were the greatest of two distinct roots of $\CS(p)=0$, then $\nu_{1-2\lambda}$ would coincide with the unit mass at zero, which explains the motivation for the adoption of $(H_2)$ in Subsection \ref{sec:intro}.

For the sake of definiteness, we recall that the parameters $k_{1-2\lambda}$ and $\gamma$, appearing in the expression of the transform $\hat{g}$ introduced in Fact \ref{fattoCS}, are related to the constants $c_1$ and $c_2$ in \eqref{limiti} through 
\[
k_{1-2\lambda}=\frac{(c_1+c_2)\pi}{2\Gamma(1-2\lambda)\sin(\pi(1-2\lambda)/2)},\qquad \gamma=\I_{\{c_1+c_2>0\}}\frac{c_2-c_1}{c_1+c_2}.
\]

\vskip 1cm

The section concludes with a couple of proofs previously omitted for not interrupting the line of reasoning.

\subsection{Proof of Facts \ref{propstazionarie} and \ref{explosion}}\label{sec:proofs}

\begin{proof}[Proof of {\rm Fact \ref{propstazionarie}}]
It is obvious that any steady state is also a limiting p.d. for $\mu_t$, as $t\to+\infty$. To prove the vice versa, split the summands in the expression of $V_t$ into two classes: the former containing those with $\tl_1$ in the respective coefficient $\tbe_j$, the latter including the remaining summands. Then $V_t=\tl_1 V_{t,l}+ \tr_1 V_{t,r}$ where $V_{t,l}$ [$V_{t,r}$, respectively] denotes the sum of the $\tbe_{j,\tnu_t}X_j/\tl_1$ [$\tbe_{j,\tnu_t}X_j/\tr_1$, respectively] associated with the summands in the first [second, respectively] class. It is an easy fact that if $V_t$ converges in distribution, then, conditionally on $(\tl_1,\tr_1)$, also $V_{t,l}$ and $V_{t,r}$ converge in distribution and the three limiting laws must coincide. Thus, denoting this common law by $\mu_\infty$, one gets $\mu_\infty=Q^+(\mu_\infty)$, where $Q^+$ is the same operator as in \eqref{C}, stating that $\mu_\infty$ is a fixed point of such an operator or, equivalently, that $\mu_\infty$ is a steady state.
\end{proof}

\begin{proof}[Proof of {\rm Fact \ref{explosion}}]
As far as (i) is concerned, it follows from the assumptions combined with the part of Fact \ref{fattoCS} concerning the case in which $c_1=c_2=0$. To deal with point (ii), for the sake of expository clarity, introduce the space
\[
D:=\Big([0,+\infty)\Big)^3\times \R\times \Big([0,1]\Big)^2\times [0,1]
\]
together with its coordinate variables
\[
(Z,Z_1,Z_2), S_\a,(\tl,\tr), U.
\]
Moreover, let $Pr$ be a p.m. on the Borel subsets $\CB(D)$ of $D$ which makes the coordinates stochastically independent and having the following additional distributional properties:
\begin{itemize}
\item[$(i_1)$] $Z,Z_1,Z_2$ are i.i.d. non-negative random numbers having mean equal to 1
\item[$(i_2)$] $S_\a$ has a non-degenerate stable law of exponent $\a$
\item[$(i_3)$] $(\tl,\tr)$ is a random vector distributed according to $\tau$
\item[$(i_4)$] $U$ is a random number uniformly distributed on $(0,1)$
\item[$(i_5)$] the p.d. of $Z$ is the same that the one of $U^{\CS(\a)}(\tl^\a Z_1+\tr^\a Z_2)$.
\end{itemize}
Then, defining $G_\infty$ to be the p.d. function of $S_\a Z^{1/\a}$, one can use Theorem 2.2 in Bassetti and Ladelli (2012) to state that $G_t(x):=\mu_t(-\infty,x e^{t\CS(\a)/\a}]\to G_\infty(x)$, as $t\to+\infty$, at each continuity point $x$ of $G_\infty$. Therefore, for every $a>0$,
\[
\begin{split}
\mu_t((-a,a])=G_t(a e^{-t\CS(\a)/a})-G_t(-a e^{-t\CS(\a)/a})\leq G_t(\veps)-G_t(-\veps-0)
\end{split}
\]  
holds true for every $\veps>0$ whenever $t\geq (\a/\CS(\a))\log(a/\veps)\vee 0$. Moreover, 
\[
\limsup_{t\to+\infty}[G_t(\veps)-G_t(-\veps-0)]\leq G_\infty(\veps)-G_\infty(-\veps-0)
\]
and hence 
\begin{equation}\label{limsup}
\limsup_{t\to+\infty}\mu_t((-a,a))\leq Pr\{S_\a Z^{1/\a}=0\}.
\end{equation}
Since every non-degenerate stable law is absolutely continuous, the probability of $\{S_\a Z^{1/\a}=0\}$ is equal to the probability of $\{Z=0\}$. Now, from  $(i_5)$, it follows that
\[
Pr\{Z=0\}=Pr\{\tl^\a Z_1+\tr^\a Z_2=0\}.
\] 
Since $Pr\{\tl=0\}=\Pr\{\tr=0\}=0$ from \eqref{conttau}, and $Z_1$ and $Z_2$ are non-negative random numbers, one can write $Pr\{\tl^\a Z_1+\tr^\a Z_2=0\}=Pr\{Z_1=0,Z_2=0\}$, and then, combination of the previous equalities with $(i_1)$ entails 
\[
Pr\{Z=0\}=Pr\{Z=0\}^2
\]
and $Pr\{Z=0\}=0$ since $Z$ has mean equal to 1. At this stage, to complete the proof, it is enough to recall \eqref{limsup}.
\end{proof}

\section*{Acknowledgements}
We are very grateful for the useful comments we have received from an anonymous referee. We also wish to thank Giuseppe Toscani for helping us to orient ourselves in the econophysical literature, and Gianluca Cassese for many useful suggestions which certainly have improved the quality of the paper.

\end{document}